\documentclass[a4paper]{article}

\usepackage{amsmath, amssymb, amsthm}
\usepackage{lineno,hyperref}
\usepackage{enumerate}

\usepackage{mathdots}

\newtheorem{theorem}{Theorem}[section]

\newtheorem{cor}[theorem]{Corollary}

\newtheorem{proposition}[theorem]{Proposition}
\newtheorem{lem}{Lemma}[section]
\newtheorem{remark}{Remark}[section]
\newtheorem{definition}{Definition}[section]
\newtheorem{example}{Example}

\allowdisplaybreaks

\begin{document}
	\title{Solutions of the tt*-equations constructed from the \(({\rm SU}_2)_k \)-fusion ring, and Smyth potentials}
	\author{Tadashi Udagawa}
	\date{}
	\maketitle

	\begin{abstract}
		Cecotti and Vafa (\cite{CV1991}) introduced the tt*-equation (topological-anti-topological fusion equation), whose solutions describe massive deformations of supersymmetric conformal field theories. We describe some solutions of the tt*-equation constructed from the \(({\rm SU}_2)_k \)-fusion algebra. The idea of the construction is due to Cecotti and Vafa, but we give a precise mathematical formulation and a description of the ``holomorphic data'' corresponding to the solutions by using the DPW method. Furthermore, we give a relation between the solutions and the representations of \({\rm SU}_2 \). As a special case, we consider the solutions corresponding to the supersymmetric \(A_k \)-minimal model.
	\end{abstract}

	\section{Introduction}
	Cecotti and Vafa (\cite{CV1991}) introduced the tt*-equation (topological-antitopological fusion equation), whose solution describes massive deformations of supersymmetric conformal field theories. From a mathematical point of view, Dubrovin (\cite{D1993}) observed that the tt* equation is an example of an integrable p.d.e., and that solutions give examples of pluriharmonic maps into the symmetric space \({\rm GL}_n \mathbb{R}/{\rm O}_n \) \((n \in \mathbb{N})\). Based on the expected properties of solutions, Cecotti and Vafa proposed an approach to the classification of quantum field theories (\cite{CV1991}, \cite{CV1993}). However, the equation is highly nonlinear, and had been solved only in very special cases.  The radial sinh-Gordon equation
	\begin{equation}
		w_{t \overline{t}} = 2\sinh{(2w)}, \ \ \ w = w(|t|),\ t \in \mathbb{C}^*, \nonumber
	\end{equation}
	was one such case; it was known from work of McCoy-Tracy-Wu (\cite{MTW1977}) that global solutions on \(\mathbb{C}^* \) are in 1:1 correspondence with real numbers \(m \in [-1,1] \), the correspondence being given by
	\begin{equation}
		w \sim m \log{|t|} \ \ \ {\rm as}\ \ t \rightarrow 0. \nonumber
	\end{equation}
	Much later, the radial Toda equation of ``tt* type'', or tt*-Toda equation, was solved by Guest-Its-Lin (\cite{GIL20151}, \cite{GIL20152}, \cite{GIL2020}), who gave a similar description of all global solutions on \(\mathbb{C}^* \).  This series of examples includes the radial sinh-Gordon equation as a special case. \\
	
	The purpose of this paper is to give another series of examples.  Here, from the p.d.e. point of view, a solution may be constructed directly from a finite number of solutions to the radial sinh-Gordon equation, but the construction itself is quite different from the case of the tt*-Toda equation.  It involves the \(({\rm SU}_2)_k\)-fusion algebra, or Verlinde algebra, an object which has a prominent role in conformal field theory, and hence a natural physical origin. The idea of the construction is due to Cecotti and Vafa, and was sketched (in the language of conformal field theory) in \cite{CV1991}. \\
	
	Our first main result is a precise mathematical formulation of this version of the tt* equation and its solutions (section 2).  Our second main result is a description of the ``holomorphic data'' corresponding to the solutions (section 3). Physically, this data can be regarded as the chiral data of the field theory. Mathematically, it is the generalized Weierstrass data, or DPW data, of the corresponding harmonic map into the symmetric space isomorphic to \({\rm GL}_{k+1} \mathbb{R}/{\rm O}_{k+1} \).  Although this data was not considered by Dubrovin or by Cecotti and Vafa, its role in the context of harmonic map theory is well known to differential geometers, and it can be expected to play an important role in describing the solutions. We shall show that this is indeed the case. \\
	
	Our third main result (section 4) makes explicit use of the representations of \({\rm SU}_2 \) which are the building blocks of the fusion algebra. We show that a natural equivalence relation on such representations corresponds to an equally natural notion of gauge equivalence on harmonic maps. In this section, the corresponding holomorphic data are characterized by nonzero integers. For the tt*-Toda equation, the relation between holomorphic data characterized by nonzero integers and positive energy representations was investigated by M. Guest and T. Otofuji \cite{GO2022}. In our case, we consider the analogue of this relation. As a special case, we also consider the tt*-equation corresponding to the supersymmetric \(A_k \) minimal model. In \cite{CV1991}, Cecotti and Vafa considered this case.
	\\
	
	A similar analysis would be possible for the \(({\rm SU}_n)_k \)-fusion algebra, for any \(n \ge 2 \), and this more general situation was also considered by Cecotti and Vafa (\cite{CV1991}). A solution of the equation in this case corresponds to a finite number of solutions of the tt*-Toda equation, the case \(k=1 \) being the tt*-Toda equation itself. In combination with the results of \cite{GIL20151} and \cite{GIL20152} on the  tt*-Toda equation, our method leads to a description of all (global) solutions of the tt* equation arising from the \(({\rm SU}_n)_k\)-fusion algebra.  We have restricted ourselves to the case \(n = 2 \) in this article, partly for simplicity of exposition, and partly because the harmonic maps involved are of particular interest in differential geometry, as they correspond to (certain) surfaces of constant mean curvature in Minkowski space.

	\section{The tt*-equation and the \(({\rm SU}_2)_k \)-fusion ring}
	In this section, we define the tt*-equation on a flat vector bundle and we review a type of tt*-equation constructed from the \(({\rm SU}_2)_k \)-fusion ring.
	\subsection{The tt*-equations on vector bundles}
	Let \(\Sigma \) be a Riemann surface with the local coordinate \(t \) and \(E \) a holomorphic vector bundle over \(\Sigma \) with a holomorphic structure \(\overline{\partial}_E \), \(\eta \) a nondegenerate holomorphic symmetric bilinear form on \(E \) and \(g \) a Hermitian metric on \(E \). We define a conjugate-linear map \(\kappa : \Gamma(E) \rightarrow \Gamma(E) \) by \(g(a,b) = \eta(\kappa(a),b)\) for \(a,b \in \Gamma(E) \). We define the tt*-equation for the metric \(g \) on \(E \) as follows. \vspace{3mm}
	
	\begin{definition} \label{def2.1}
		Let \(Cdt \) be an \({\rm End}(E)\)-valued holomorphic 1-form. We assume that \(\kappa \) satisfies the ``reality constraint'' \(\kappa^2 = Id_E \) and \(C \) is self-adjoint with respect to \(\eta \). We define a family of connection \(\nabla^{\lambda }\ (\lambda \in S^1) \) on \(E \) by
		\begin{equation}
			\nabla^{\lambda} = \partial_E^g + \overline{\partial}_E + \lambda^{-1}Cdt + \lambda C^{\dagger_g}d\overline{t}, \nonumber
		\end{equation}
		where \(\partial_E^g \) is the Chern connection and \(C^{\dagger_g } \) is the adjoint of \(C \) with respect to \(g \). We say \(g \) satisfies the tt*-equation if \(\nabla^{\lambda} \) is flat for all \(\lambda \in S^1 \).
	\end{definition} \vspace{3mm}
	
	We omit until later the homogeneity condition (Euler condition) which is usually given in the context of Frobenius manifolds (see \cite{D1993}).  \vspace{3mm}
	
	\begin{example}
		Let \(E = \mathbb{C}^* \times \mathbb{C}^{n+1}\ (n \in \mathbb{N}) \) be a trivial vector bundle and \(\{e_j \}_{j=0}^{n} \) be the standard frame of \(E \). We consider the tt*-equation for \(\eta(e_j,e_l) = \delta_{j,n-l} \) and
		\begin{equation}
			Cdt (e_0,\cdots,e_n) = (e_0,\cdots,e_n) \left(\begin{array}{cccc}
				0 & & & 1 \\
				1 & \ddots & &   \\
				& \ddots & \ddots &  \\
				& & 1 & 0 
			\end{array} \right)dt,\ \ \ t \in \mathbb{C}^*. \nonumber
		\end{equation}
		We put \(g(e_j,e_l) = e^{w_j} \delta_{jl} \), where \(w_j: \mathbb{C}^* \rightarrow \mathbb{R} \), then from the reality constraint the \(\{w_j\}_{j=0}^{n} \) satisfy
		\begin{equation}
			w_j + w_{n-j} = 0,\ \ \ j = 0,\cdots, n. \nonumber
		\end{equation}
		From the tt*-equation, the \(\{w_j \}_{j=0}^{n} \) satisfy
		\begin{equation}
			(w_j)_{t \overline{t}} = e^{w_j-w_{j-1} } - e^{w_{j+1} - w_j },\ \ \ j = 0,\cdots,n, \nonumber
		\end{equation}
		where \(w_{-1} = w_n, w_{n+1} = w_0 \). This may be called the tt*-Toda equation.
		In physics, the tt*-Toda equation was introduced by Cecotti and Vafa in \cite{CV1991}. The global solutions of the tt*-Toda equation on \(\mathbb{C}^* \) were investigated by Guest, Its and Lin (see \cite{GIL20151}, \cite{GIL20152}, \cite{GIL2020}). \\
		\qed
	\end{example}  \vspace{3mm}
	In this situation, the homogeneity condition amounts to the radial condition \(w_j = w_j(|t|) \).

	\subsection{The \(({\rm SU}_2)_k \)-fusion ring}
	In this paper, we consider another example of the tt*-equation given by the algebraic structure of the \(({\rm SU}_2 )_k \)-fusion ring. This type of tt*-equation was introduced by Cecotti and Vafa. In \cite{CV1991}, they constructed a solution of the tt*-equation written by using solutions of the tt*-Toda equation of Example 1 with \(n = 1 \) (i.e. the sinh-Gordon equation). In this section, we explain this solution precisely and we describe the global solutions of the tt*-equation constructed from the \(({\rm SU}_2)_k \)-fusion ring. \vspace{3mm}
	
	For \(k \in \mathbb{N} \), the \(({\rm SU}_2)_k \)-fusion ring can be described by the polynomial ring modulo a certain ideal
	\begin{equation}
		R_k = \mathbb{C}[X]/\left(dT_{k+2}/dX \right), \nonumber
	\end{equation}
	where \(T_{k+2} \) is the first kind of Chebyshev polynomial given by
	\begin{equation}
		T_{k+2}(\cos{\theta}) = \cos{((k+2)\theta)}. \nonumber
	\end{equation}
	We denote the equivalence class of a polynomial \(p(X)\) by \([p(X)] \). Next, we define a nondegenerate bilinear form and an endomorphism \(C \) on \(R_k \).  \vspace{3mm}
	
	\begin{definition}
		We define a nondegenerate bilinear form \(<\cdot,\cdot> \) on \(R_k \) by the Grothendieck residue
		\begin{equation}
			<[a],[b]> = {\rm Res}_{T_{k+2} }(ab) = \frac{1}{2 i \pi} \int_{\gamma } \frac{a(X)b(X)}{\frac{dT_{k+2}}{dX} }dX,\ \ \ a,b \in \mathbb{C}[X], \nonumber
		\end{equation}
		where \(\gamma \) is a closed curve containing all roots of \(dT_{k+2}/dX \).
	\end{definition}
	\vspace{3mm}
	
	\begin{definition}
		We define an endomorphism \(C:R_k \rightarrow R_k \) as the map induced by multiplication by \(T_{k+2} \) i.e.
		\begin{equation}
			C([a]) = \left[T_{k+2} \cdot a \right] \in R_k. \nonumber
		\end{equation}
	\end{definition}
	\vspace{3mm}
	
	\begin{remark}
		In physics, the construction of \(<\cdot,\cdot> \) and \(C \) above arise in the Landau-Ginzburg theory. As a generalization of the variation of Hodge structure, \(C \) and \(<\cdot,\cdot> \) are the analogue of the Kodaira-Spencer map and the Poincar\'{e} pairing respectively (see \cite{C1991}). \\
		\qed
	\end{remark}  \vspace{3mm}
	
	Let \(E = \mathbb{C}^* \times R_k \) be the trivial holomorphic vector bundle over \(\mathbb{C}^* \) with the holomorphic structure \(\overline{\partial}_E = \overline{\partial} \), \(\eta \) the nondegenerate holomorphic bilinear form on \(E \) induced by \(<\cdot,\cdot> \) and \(Cdt\ (t \in \mathbb{C}^*) \) an \({\rm End}(E) \)-valued holomorphic 1-form on \(\mathbb{C}^* \). We consider the tt*-equation for \((E,\eta,C) \). We have the following Lemma. \vspace{3mm}
	
	\begin{lem} \label{lem2.1}
		The fusion ring \(R_k \) is a \(k+1 \) dimensional \(\mathbb{C} \)-module with a basis \(([1],[X],\cdots,[X^k] ) \) i.e. \(E \) is a rank \(k+1 \) vector bundle.
	\end{lem}
	\begin{proof}
		From the definition, the degree of \(T_{k+2} \) is \(k+2 \) and then the degree of \(dT_{k+2}/dX \) is \(k+1 \). Hence \(([1],[X],\cdots,[X^k] ) \) is a basis of \(R_k \). \\
	\end{proof}  \vspace{3mm}
	
	The following properties are well-known (\cite{CV1991}).  \vspace{3mm}
	
	\begin{lem}\label{lem2.3}
		The roots \(\{X_j \}_{j \in J} \) of \(dT_{k+2}/dX \) are given by
		\begin{equation}
			X_j = \cos{\left(\frac{j}{k+2} \pi \right) },\ \ \ i \in J = \{1,\cdots,k+1 \}. \nonumber
		\end{equation}
		We have
		\begin{equation}
			<[a],[b]> = \sum_{j=1}^{k+1} a(X_j) b(X_j) \left(\frac{d^2T_{k+2}}{dX^2}(X_j) \right)^{-1},\ \ \ a,b \in \mathbb{C}[X]. \nonumber
		\end{equation}
	\end{lem}
	\begin{proof}
		From the definition of \(T_{k+2} \) with \(X = \cos{\theta} \), we have
		\begin{equation}
			\frac{dT_{k+2}}{dX}(\cos{\theta} ) = (k+2)\frac{\sin{((k+2)\theta)}}{\sin{\theta} }. \nonumber
		\end{equation}
		Thus, roots of \(dT_{k+2}/dX \) are given by \(\cos{(j \pi/(k+2)) } \) for \(j = 1,\cdots,k+1\). Since \(dT_{k+2}/dX = 2^{k+1} \prod_{r=1}^{k+1} (X-X_r) \), we have
		\begin{equation}
			\frac{d^2T_{k+2}}{dX^2}(X_j) = 2^{k+1} \prod_{r \in J \backslash \{j \}} (X_j-X_r). \nonumber
		\end{equation}
		From the residue Theorem, we obtain
		\begin{align}
			<[a],[b]> &= \sum_{j=1}^{k+1} \lim_{X \rightarrow X_j} (X-X_j) \frac{a(X)b(X)}{\frac{dT_{k+2}}{dX} } \nonumber \\
			&= \sum_{j=1}^{k+1} a(X_j)b(X_j) \left(\frac{d^2T_{k+2}}{dX^2}(X_j) \right)^{-1}. \nonumber
		\end{align}
	\end{proof}
	\vspace{3mm}
	
	\begin{lem} \label{lem2.2}
		Let \(a,b \in \mathbb{C}[X] \) and \(\{X_j \}_{j \in J} \) be the roots of \(dT_{k+2}/dX \), then
		\begin{equation}
			[a] = [b]\ \ \ \Leftrightarrow \ \ \ a(X_j) = b(X_j)\ \ {\rm for\ all}\ j \in J = \{1,\cdots,k+1 \}. \nonumber
		\end{equation}
	\end{lem}
	\begin{proof}
		If \([a] = [b] \), then \(a-b \in dT_{k+2}/dX \cdot \mathbb{C}[X] \). Thus, we obtain \(a(X_j) = b(X_j) \) for all \(j \in J \). Conversely, suppose that \(a(X_j) = b(X_j) \) for all \(j \in J \). 
		Since \(a-b = A(X)\prod_{j \in J}(X-X_j) = A(X) dT_{k+2}/dX \) for some \(A \in \mathbb{C}[X] \), we obtain \([a] = [b] \in R_k \). \\
	\end{proof}  \vspace{3mm}
	
	In this section, we seek a Hermitian metric \(g \) satisfying the tt*-equation for \((E,\eta,C) \) constructed above. In general, it is very difficult to solve the tt*-equation. If \(g \) is diagonal with respect to some basis, the calculation become more feasible. Thus, it is important to choose a suitable basis. \vspace{3mm}
	
	If we use a basis \(([1],[X],\cdots,[X^k]) \) of \(R_k \), then the representation matrix of \(\eta \) is not simple and the representation matrix of \(g \) is not diagonal. For example if we consider the case \(k = 2 \) then
	\begin{equation}
		(<[X^j],[X^l]>)_{0 \le j,l \le 2} = \frac{1}{8} \left(\begin{array}{ccc}
			0 & 0 & 2 \\
			0 & 2 & 0 \\
			2 & 0 & 1
		\end{array} \right). \nonumber
	\end{equation}
	Suppose that \(g \) is diagonal with respect to \(([1],[X],[X^2]) \), then \(g, \eta \) do not satisfy the reality constraint. Thus, we can not assume that \(g \) is diagonal.  \vspace{3mm}
	
	\begin{remark}
		In \cite{CV1991}, Cecotti and Vafa used the ``point basis'' \(\{[L_j] \}_{j = 1}^{k+1} \) defined by \(L_j(X_l) = \delta_{jl} \). Then the matrix of \(C \) with respect to \(\{[L_j] \}_{j=1}^{k+1} \) is given by
		\begin{equation}
			C([L_1],\cdots,[L_{k+1}] ) = ([L_1],\cdots,[L_{k+1}] ) \left(\begin{array}{ccc}
				-1 & & \\
				& \ddots & \\
				& & (-1)^{k+1}
			\end{array} \right), \nonumber
		\end{equation}
		and we have
		\begin{equation}
			<[L_j],[L_l]> = (-1)^{j+1} \frac{\sin{\left(\frac{j}{k+2}\pi \right)}}{k+2}  \frac{\sin{\left(\frac{l}{k+2}\pi \right)} }{k+2 } \delta_{jl}. \nonumber
		\end{equation}
		If we assume that \(g \) is diagonal with respect to \(\{L_j \}_{j=1}^{k+1} \), then  the reality constraint implies \(g \) is constant i.e. we obtain only a trivial solution of the tt*-equation. Thus, a non-trivial solution \(g \) is not diagonal with respect to \(\{L_j \}_{j=1}^{k+1} \).
		\qed
	\end{remark}  \vspace{3mm}
	
	In \cite{CV1991}, Cecotti and Vafa gave a ``physical argument'' for solving the tt*-equation constructed from the \(({\rm SU}_2)_k \)-fusion ring. In this paper, we choose a special basis of \(R_k \) and we simplify the construction of \(g \).  \vspace{3mm}
	
	\begin{lem}
		Let \(U_n\ (n \in \mathbb{Z}_{\ge 0}) \) be the second kind of Chebyshev polynomial defined by
		\begin{equation}
			U_n(\cos{Y}) = \frac{\sin{((n+1)Y )}}{\sin{Y}}. \nonumber
		\end{equation}
		Then \(([U_0],\cdots,[U_k] ) \) is a basis of \(R_k \).  
	\end{lem}
	\begin{proof}
		This follows immediately from the fact that \(U_n(X) = 2^n X^n + O(X^{n-1}) \). \\
	\end{proof}  \vspace{3mm}
	
	Regarding \(\{[U_j] \}_{j=0}^{k} \), we obtain the following two Lemmas.  \vspace{3mm}
	
	\begin{lem}\label{lem2.4}
		For \(k \in \mathbb{N} \), then we have
		\begin{equation}
			T_{k+2} \cdot [U_l] = - [U_{k-l}],\ \ \ l = 0,\cdots,k. \nonumber
		\end{equation}
	\end{lem}
	\begin{proof}
		From
		\begin{align}
			(T_{k+2} U_l)\left(\cos{\left(\frac{j}{k+2} \pi \right) } \right) &= \cos{(j \pi) } \frac{\sin{\left(\frac{j}{k+2}\pi (l+1) \right) } }{\sin{(\frac{j}{k+2}\pi )} } \nonumber \\
			&= (-1)^{2j+1} \frac{\sin{\left(\frac{j}{k+2}\pi (k-l+1) \right) } }{\sin{(\frac{j}{k+2}\pi )} } \nonumber \\
			&= - U_{k-l}\left(\cos{\left(\frac{j}{k+2} \pi \right) } \right),\ \ \ j = 1,\cdots,k+1. \nonumber
		\end{align}
		and Lemma \ref{lem2.2}, we obtain \([T_{k+2}] \cdot [U_l] = - [U_{k-l}] \). \\
	\end{proof}  \vspace{3mm}
	
	\begin{lem}\label{lem2.6}
		We have
		\begin{equation}
			<[U_j],[U_l]> = \frac{1}{2(k+2)} \delta_{j,k-l}. \nonumber
		\end{equation}
	\end{lem}
	\begin{proof}
		From Lemma \ref{lem2.3},
		we have
		\begin{align}
			<[U_j],[U_l]> 
			&= \frac{1}{(k+2)^2} \sum_{r=1}^{k+1} \sin{\left(\frac{j+1}{k+2} r\pi \right)} \sin{\left(\frac{k-l+1}{k+2} r\pi \right)}. \nonumber
		\end{align}
		Since
		\begin{align}
			&\sin{\left(\frac{j+1}{k+2} r\pi \right)} \sin{\left(\frac{k-l+1}{k+2} r\pi \right)} \nonumber \\
			&\hspace{2cm} + \sin{\left(\frac{j+1}{k+2} (k+2-r)\pi \right)} \sin{\left(\frac{k-l+1}{k+2} (k+2-r)\pi \right)} \nonumber \\
			&= \left\{1+(-1)^{j+k-l } \right\} \sin{\left(\frac{j+1}{k+2} r\pi \right)} \sin{\left(\frac{k-l+1}{k+2} r\pi \right)}, \nonumber
		\end{align}
		for \(r = 1,\cdots,[r/2] \) and
		\begin{equation}
			\sum_{r=1}^{m} e^{i\theta r} = \frac{e^{i(m+1)\theta}-e^{i\theta} }{e^{i\theta} - 1},\ \ \ m \in \mathbb{N},\ e^{i\theta} \in S^1 \backslash \{1 \}, \nonumber
		\end{equation}
		we obtain
		\begin{equation}
			<[U_j],[U_l]> = \frac{1}{2(k+2)} \delta_{j,k-l}. \nonumber
		\end{equation}
	\end{proof}  \vspace{3mm}
	
	Combining Lemma \ref{lem2.4} and \ref{lem2.6}, we obtain
	\begin{equation}
		C([U_0],\cdots,[U_k] ) = ([U_0],\cdots,[U_k] ) \left(\begin{array}{ccc}
			& & -1 \\
			& \iddots & \\
			-1 & &
		\end{array} \right). \nonumber
	\end{equation}
	and \(<[U_j],[U_l]> = \frac{1}{2(k+2)} \delta_{j,k-l} \). Hence the matrices of \(C \) and \(\eta \) with respect to \(([U_0],\cdots,[U_k] ) \) are very simple. \vspace{3mm}
	
	We define a holomorphic frame of \(E \) by
	\begin{equation}
		\tau_j : \mathbb{C}^* \rightarrow E : t \mapsto (t,\sqrt{2(k+2)} [U_j]),\ \ \ j = 0, \cdots,k. \nonumber
	\end{equation}
	In this section, we solve the tt*-equation for \((E,\eta,C) \) under the assumption \(g \) is diagonal with respect to \((\tau_0,\cdots,\tau_k) \).  \vspace{3mm}
	
	\begin{remark}\label{rem2.3}
		Under the assumption above, \(g \) satisfies \(g\left((-1)^j \tau_j, (-1)^l \tau_l \right) = g(\tau_j,\tau_l) \) for any \(j,l = 0\cdots,k \).
		This symmetry is called the ``\(\mathbb{Z}_2 \)-symmetry'' and induced by the symmetries \(T_{k+2}(-X) = (-1)^{k+2} T_{k+2}(X) \) and \(U_j(-X) = (-1)^j U_j(-X) \) of the Chebyshev polynomials in physics \cite{CV1991}, \cite{CV19932}, \cite{CV1993}. In this paper, we assume the ``\(\mathbb{Z}_{k+1} \)-symmetry'', the generalization of \(\mathbb{Z}_2 \)-symmetry, i.e. \(g(e^{\sqrt{-1} \frac{2\pi}{k+1}j} \tau_j,e^{\sqrt{-1} \frac{2\pi}{k+1}l } \tau_l) = g(\tau_j,\tau_l ) \) for any \(j,l=0,\cdots,k \). \\
		\qed
	\end{remark}  \vspace{3mm}
	
	We obtain the following Proposition.  \vspace{3mm}
	
	\begin{proposition} \label{prop2.1}
		Let \(g \) be a Hermitian metric on \(E \) such that \(g \) is diagonal with respect to \((\tau_0,\cdots,\tau_k) \). \\
		Then \(g \) satisfies the tt*-equation if \(w_j = \log{g(\tau_j,\tau_j)} \ (j=0,\cdots,k) \) are solutions of 
		\begin{equation}
			(w_j)_{t\overline{t}} = e^{w_j-w_{k-j}} - e^{w_{k-j}-w_j}, \nonumber
		\end{equation}
		such that \(w_j = -w_{k-j} \).
	\end{proposition}
	\begin{proof}
		We have
		\begin{align}
			&\nabla^{\lambda} (\tau_0,\cdots,\tau_k) \nonumber \\
			&\  = (\tau_0,\cdots,\tau_k ) \left\{
			\left(\begin{array}{ccc}
				(w_0)_t & & \\
				& \ddots & \\
				& & (w_k)_t
			\end{array} \right)dt + \left(\begin{array}{ccc}
				& & -\lambda^{-1} \\
				& \iddots & \\
				-\lambda^{-1} & &
			\end{array} \right)dt \right. \nonumber \\
			&\left. \hspace{5cm} + \left(\begin{array}{ccc}
				& & -\lambda e^{w_k - w_0 } \\
				& \iddots & \\
				-\lambda e^{w_0 - w_k } & & 
			\end{array} \right)d\overline{t}
			\right\}. \nonumber
		\end{align}
		The flatness condition gives the stated result. \\
	\end{proof}  \vspace{3mm}
	
	Thus, a solution of the tt*-equation constructed from the \(({\rm SU}_2)_k \)-fusion ring corresponds to \(k+1 \)-solutions of the sinh-Gordon equation. \\
	
	Let \(w_j:\mathbb{C}^* \rightarrow \mathbb{R}\ (j=0,\cdots,k) \) be solutions of
	\begin{equation}
		(w_j )_{t \overline{t} } = e^{w_j-w_{k-j} } - e^{w_{k-j} - w_j }, \tag{A} \label{A}
	\end{equation}
	with the conditions
	\begin{enumerate}
		\item [(1)] \(w_j = w_j(|t|) \) (radial condition), \vspace{1mm}
		\item [(2)] \(w_j + w_{k-j} = 0 \) (anti-symmetry condition).
	\end{enumerate}
	From the condition (2), the equation (\ref{A}) is equivalent to the sinh-Gordon equation
	\begin{equation}
		(w_j )_{t \overline{t} } = e^{2w_j}  - e^{- 2w_j }. \nonumber
	\end{equation}
	
	From Proposition \ref{prop2.1}, the tt*-equation constructed from the \(({\rm SU}_2)_k \)-fusion ring is equivalent to \((k+1)/2 \) (if \(k+1 \) is even) or \(k/2 \) (if \(k+1 \) is odd) copies of the sinh-Gordon equation. \\
	
	It is well-known (see \cite{FIKN2006}, \cite{GIL20151}, \cite{MTW1977}) that the global solutions of the sinh-Gordon equation are characterized by their asymptotic behaviour at \(t = 0 \)
	\begin{equation}
		w_j \sim m_j \log{|t|}\ \ \ {\rm as}\ \ t \rightarrow 0, \ \ \ -1 \le m_j \le 1. \nonumber
	\end{equation}
	Thus, we obtain the following relation between solutions of (A) and a subset of \(\mathbb{R}^{k+1} \):  \vspace{3mm}
	
	\begin{proposition}[\cite{GIL20151}]
		There is a one-to-one correspondence between
		\begin{enumerate}
			\item [(a)] solutions of (A), \vspace{1mm}
			
			\item [(b)] \(k+1 \)-tuples \((m_0,\cdots,m_k) \in \mathbb{R}^{k+1} \) with \(-1 \le m_j \le 1,\ m_j+m_{k-j}=0 \).
		\end{enumerate}
		We call \((m_0,\cdots,m_k) \) asymptotic data.
	\end{proposition}  \vspace{3mm}
	
	For \(k \in \mathbb{N} \), the solution of the tt*-equation of (A) consists of \(k+1 \)-solutions of the sinh-Gordon equation. In the following section, we investigate the relation between the tt*-equation constructed from the \(({\rm SU}_2)_k \)-fusion ring and holomorphic data (holomorphic potentials) in terms of the DPW method.

	\section{The \(({\rm SU}_2)_k \)-fusion ring and holomorphic data}
	In this section, we characterize the solutions of (A) by holomorphic 1-forms, or, equivalently, by \(k+1 \)-tuples \((l_0,\cdots,l_k) \). More precisely, we construct a solution of (A) from a matrix-valued 1-form
	\begin{equation}
		\frac{1}{\lambda} \left(\begin{array}{ccc}
			& & z^{l_0} \\
			& \iddots & \\
			z^{l_k} & & 
		\end{array} \right)dz, \ \ \ \lambda \in S^1,\ z \in \mathbb{C}-(-\infty,0],\ l_0,\cdots,l_k \in [-1,\infty), \nonumber
	\end{equation}
	by using the DPW method. In general, the reality constraint \(\kappa^2 = Id \) induces a certain symmetry on \(\nabla^{\lambda } \) in the tt*-structure \((E,\eta,C,g,\nabla^{\lambda}) \) (see \cite{D1993}). From this symmetry, we give a real form on \({\rm Aut}(E) \) then the DPW method can be applied to describe the tt*-structure. \vspace{1mm} \\
	
	First, we make a gauge transformation of the connection \(\nabla^{\lambda} \) of Definition 2.1.
	\begin{proposition}
		Let \((E,\eta,g,C,\nabla^{\lambda}) \) be the tt*-structure in Proposition \ref{prop2.1}. Put
		\begin{equation}
			(\tilde{\tau}_0,\cdots,\tilde{\tau}_k) = (\tau_0,\cdots,\tau_k) \left(\begin{array}{ccc}
				e^{\frac{w_k}{2}} & & \\
				& \ddots & \\
				& & e^{\frac{w_0}{2} }
			\end{array} \right), \nonumber
		\end{equation}
		and define \(\alpha \) by \(\nabla^{\lambda}(\tilde{\tau}_0,\cdots,\tilde{\tau}_k) = (\tilde{\tau}_0,\cdots,\tilde{\tau}_k) \alpha \). Then we have
		\begin{align}
			\alpha &= \left(\begin{array}{ccc}
				\frac{1}{2}(w_0)_t & & -\lambda^{-1} e^{\frac{1}{2}(w_0-w_k) } \\
				& \ddots & \\
				-\lambda^{-1} e^{\frac{1}{2}(w_k-w_0) } & & \frac{1}{2} (w_k)_t 
			\end{array} \right)dt \nonumber \\
			&\hspace{3cm}+ \left(\begin{array}{ccc}
				\frac{1}{2} (w_0)_{\overline{t} } & & -\lambda e^{\frac{1}{2}(w_k-w_0 ) } \\
				& \ddots & \\
				-\lambda e^{\frac{1}{2}(w_0-w_k ) } & & \frac{1}{2} (w_k)_{\overline{t} }
			\end{array} \right)d\overline{t}. \nonumber
		\end{align}
	\end{proposition}
	\begin{proof}
		Direct calculation. \\
	\end{proof}  \vspace{3mm}
	
	To apply the DPW method (loop group method) we regard \(\alpha \) as a loop algebra-valued 1-form. We note first that it has the symmetries
	\begin{equation}
		c(\alpha(\lambda)) = \alpha(1/\overline{\lambda}),\ \ \ \sigma(\alpha(\lambda) ) = \alpha(-\lambda), \nonumber
	\end{equation}
	where
	\begin{equation}
		c(A) = \Delta \overline{A} \Delta,\ \ \ \sigma(A) = -\Delta A^t \Delta,\ \ \ \Delta = \left(\begin{array}{ccc}
			& & 1 \\
			& \iddots & \\
			1 & &
		\end{array} \right), \nonumber
	\end{equation}
	for \(A \in M_{k+1} \mathbb{C} \). As shown in \cite{D1993}, \cite{GIL20151}, \cite{GIL20152}, \cite{GIL2020}, these symmetries imply that a solution of (A) corresponds to a harmonic map into the symmetric space \({\rm GL}_{k+1} \mathbb{R}/{\rm O}_{k+1} \).  \vspace{3mm}
	
	\begin{remark}
		The connection form \(\alpha \) also has the symmetry
		\begin{equation}
			\tau(\alpha(\lambda)) = \alpha((-1)^{k+2} \lambda), \nonumber
		\end{equation}
		where
		\begin{equation}
			\tau(A) = \left(\begin{array}{ccc}
				1 & & \\
				& \ddots & \\
				& & (-1)^k
			\end{array} \right) A \left(\begin{array}{ccc}
				1 & & \\
				& \ddots & \\
				& & (-1)^k
			\end{array} \right), \nonumber
		\end{equation}
		for \(A \in {\rm M}_{k+1} \mathbb{C} \). From the view point of physics, this symmetry is given by the ``\(\mathbb{Z}_2 \)-symmetry'' in Remark \ref{rem2.3}. \\
		\qed
	\end{remark}  \vspace{3mm}
	
	The DPW construction for a semisimple Lie group \(G \subset {\rm GL}_n \mathbb{C}\ (n \in \mathbb{N}) \) can be summarized as follows.
	\begin{enumerate}
		\item [(a)] We solve \(d\phi = \phi \xi \) with the initial condition \(\phi(z_0) = \phi_0 \in \Lambda G^{\mathbb{C}}, z_0 \in \mathbb{C}-(-\infty,0] \). \vspace{2mm}
		\item [(b)] We split \(\phi = \phi_{\mathbb{R}} B \) by the Iwasawa factorization for \(\Lambda G^{\mathbb{C}} \) on a neighbourhood \(\mathcal{U} \subset \mathbb{C}^* \) of \(z_0 \), where \(\phi_{\mathbb{R}} \in \Lambda G, B \in \Lambda^+ G^{\mathbb{C} } \). \vspace{2mm}
		\item [(c)] We define \(\alpha = \phi_{\mathbb{R}}^{-1} d\phi_{\mathbb{R}} \) on \(\mathcal{U} \), then \(\alpha \) satisfies the zero-curvature equation.
	\end{enumerate}
	In our case, we consider the DPW construction for the real form \(G \) of \({\rm GL}_{k+1} \mathbb{C} \) with respect to the involution induced by \(c \). For example, for the case \(k=1 \) we choose \(G = {\rm SU}_{1,1} \).  \vspace{3mm}
	
	\begin{remark}
		In general, the Iwasawa factorization of \(\phi(z) \) is not defined for all \(z \in \mathbb{C}^* \). However, for the case \(G = {\rm SU}_{1,1} \) and Smyth potentials \(\xi \) we can choose \(\phi_0, z_0 \) so that \(\mathcal{U} = \mathbb{C}^* \) if \(l_0,l_1 \in \mathbb{Z}_{>-1} \) and  \(\mathcal{U} = \mathbb{C}-(-\infty,0] \) if \(l_0 \notin \mathbb{Z}_{>-1} \) or \(l_1 \notin \mathbb{Z}_{>-1} \) (see Corollary 7.5 and Corollary 8.1 of \cite{GIL2020}). In both cases, we have \((zB)|_{\lambda = 0} \rightarrow 0 \) as \(z \rightarrow 0 \) and \(B(z,\overline{z})|_{\lambda = 0} = B(|z|)|_{\lambda = 0} \). \\
		\qed
	\end{remark}

	\subsection{The case \(k= 1 \).}
	The tt*-equation constructed from the \(({\rm SU}_2)_1 \)-fusion ring is the sinh-Gordon equation itself. We follow \cite{GL2014}. Let us put \(w = w_0 = -w_1 \), then \(w \) is a solution of the sinh-Gordon equation
	\begin{equation}
		(w)_{t \overline{t}} = e^{2w} - e^{-2w}, \nonumber
	\end{equation}
	with the asymptotic behaviour \(w \sim m \log{|t|} \) as \(t \rightarrow 0 \) for \(-1 \le m \le 1 \). We recall the one-to-one correspondence between solutions of the sinh-Gordon equation and certain matrix-valued holomorphic 1-forms.  \vspace{3mm}
	
	\begin{proposition}[\cite{BI1995}, \cite{GIL2020}, \cite{GL2014}]\label{prop3.3}
		For any fixed normalization constant \(n > -2 \), there is a one-to-one correspondence between global solutions \(w \) of the sinh-Gordon equation and holomorphic 1-forms
		\begin{equation}
			\xi = \frac{1}{\lambda} \left(\begin{array}{cc}
				0 & z^{l_0 } \\
				z^{l_1 } & 0
			\end{array} \right)dz, \nonumber
		\end{equation}
		on \(\mathbb{C} - (-\infty,0] \), where \(l_0,l_1 \in \mathbb{R}_{\ge -1},\ l_0 + l_1 = n \). Here, we have
		\begin{equation}
			w \sim \frac{l_0-l_1}{n+2} \log{|t|}\ \ \ {\rm as}\ \ t \rightarrow0,\ \ {\rm where}\ \ t = \frac{2}{n+2}z^{\frac{n+2}{2} }. \nonumber
		\end{equation}
	\end{proposition}

	\subsection{The case \(k \ge 2 \).}
	The DPW method can be applied in the same way to the case \(k \ge 2 \). We begin with the holomorphic potential
	\begin{equation}
		\xi = \frac{1}{\lambda} \left(\begin{array}{ccc}
			& & z^{l_0} \\
			& \iddots & \\
			z^{l_k} & &
		\end{array} \right)dz,\ \ \ z \in \mathbb{C}-(-\infty,0],\ \ l_0,\cdots,l_k \in \mathbb{R}, \nonumber
	\end{equation}
	Let
	\begin{equation}
		\alpha_k = (w_t + \frac{1}{\lambda} W )dt + (-w_{\overline{t} } + \lambda W^t )d\overline{t},\ \ \ \lambda \in S^1, \nonumber
	\end{equation}
	where
	\begin{equation}
		w = \left(\begin{array}{ccc}
			\frac{1}{2} w_0 & &  \\
			& \ddots & \\
			& & \frac{1}{2} w_k
		\end{array} \right),\ \ \
		W = \left(\begin{array}{ccc}
			& & e^{\frac{1}{2} (w_0 - w_k) } \\
			& \iddots & \\
			e^{\frac{1}{2}(w_k - w_0) } & & 
		\end{array} \right). \nonumber
	\end{equation}
	Then the zero curvature equation \(d\alpha_k + \alpha_k \wedge \alpha_k = 0 \) is equivalent to (\ref{A}). We obtain the following generalization of Proposition \ref{prop3.3}.  \vspace{3mm}
	
	\begin{proposition}\label{prop3.4}
		For any fixed normalization constant \(n > -2 \), there is a one-to-one correspondence between global solutions \(\{w_j \}_{j=0}^{k} \) of (A) and holomprohic potentials
		\begin{equation}
			\xi = \frac{1}{\lambda} \left(\begin{array}{ccc}
				& & z^{l_0} \\
				& \iddots & \\
				z^{l_k} & &
			\end{array} \right)dz, \nonumber
		\end{equation}
		on \(\mathbb{C}-(-\infty,0]\), where \(l_0,\cdots,l_k \in \mathbb{R}_{\ge -1 },\ l_j + l_{k-j} = n \). Here, we have
		\begin{equation}
			w_j \sim \frac{l_j - l_{k-j}}{n+2} \log|t| \ \ \ {\rm as}\ \ t \rightarrow 0,\ \ {\rm where}\ \ t = \frac{2}{n+2}z^{\frac{n+2}{2} }. \nonumber
		\end{equation}
	\end{proposition}
	\begin{proof}
		We define a constant matrix \(D = (d_{jl})_{0 \le j,l \le k} \) by
		\begin{equation}
			d_{2m,l} = \left\{\begin{array}{cc}
				1 & {\rm if}\ l = m, \\
				0 & {\rm otherwise},
			\end{array}
			\right.\ \ \ 	d_{2m+1,l} = \left\{\begin{array}{cc}
				1 & {\rm if}\ l = k-m, \\
				0 & {\rm otherwise},
			\end{array}
			\right. \nonumber
		\end{equation}
		then we have
		\begin{equation}
			\xi \cdot D^{-1} = \left\{ \begin{array}{cc} \left(\begin{array}{ccc}
					\xi_0 & & \\
					& \ddots & \\
					& & \xi_{[\frac{k-1}{2}]}
				\end{array} \right) & \ \ \ {\rm if}\ k\ {\rm is\ odd}, \\
				\left(\begin{array}{cccc}
					\xi_0 & & & \\
					& \ddots & & \\
					& & \xi_{[\frac{k-1}{2}]} & \\
					& & & \lambda^{-1} z^{l_{\frac{k}{2}}}
				\end{array} \right) & \ \ \ {\rm if}\ k\ {\rm is\ even},
			\end{array} \right. \nonumber
		\end{equation}
		where
		\begin{equation}
			\xi_j = \frac{1}{\lambda} \left(\begin{array}{cc}
				0 & z^{l_j} \\
				z^{l_{k-j}} & 0
			\end{array} \right)dz,\ \ \ j = 0,\cdots,\left[\frac{k-1}{2} \right]. \nonumber
		\end{equation}
		From Proposition \ref{prop3.3}, we obtain a one-to-one correspondence between solutions of (A) and holomorphic potentials \(\xi \). \\
	\end{proof}  \vspace{3mm}
	
	Hence, we obtain a solution of the tt*-equation constructed from the \(({\rm SU}_2)_k \)-fusion ring from a holomorphic potential. We also call this \(\xi \) (or \((l_0,\cdots,l_k) \)) holomorphic data.
	Next, we return to the definition of the \(({\rm SU}_2)_k \)-fusion ring and we investigate the relation between holomorphic data and representations of \({\rm SU}_2 \).

	\section{The \(({\rm SU}_2)_k \)-fusion ring and representations of \({\rm SU}_2 \)}
	The \(({\rm SU}_2)_k \)-fusion ring can be described as a quotient \(R({\rm SU}_2)/I_k({\rm SU}_2) \) of the representation ring \(R({\rm SU}_2) \) by the ``fusion ideal'' generated by the \((k+1) \)-th symmetric power of the standard representation. In this section, we consider the projection from representations of \({\rm SU}_2 \) into the \(({\rm SU}_2)_k \)-fusion ring and we give a relation between representations of \({\rm SU}_2 \) and solutions of the tt*-equation constructed from the \(({\rm SU}_2)_k \)-fusion ring.  \vspace{2mm}
	
	Let \(\pi_n:{\rm SU}_2 \rightarrow {\rm GL}_{n+1} \mathbb{C}\ (n \in \mathbb{Z}_{\ge 0 })  \) be the \(n \)-th symmetric power \(S^n(\mathbb{C}^2 ) \) of the standard representation, then \(\{\pi_j \}_{j \in \mathbb{Z}_{\ge 0} } \) give a basis of \(R({\rm SU}_2) \). The fusion ring \(I_k({\rm SU}_2) \) is defined as the ideal generated by \(\pi_{k+1} \). Since \(R({\rm SU}_2) \) can be identified with the character ring of \({\rm SU}_2 \), we can regard \(R({\rm SU}_2) \) as the polynomial ring \(\mathbb{C}[X] \) by identifying the character \(\chi_{\pi_1} \) of \(\pi_1 \) with \(2X \).  \vspace{3mm}
	
	\begin{lem}\label{lem4.1}
		Under the identification above, we have
		\begin{equation}
			\chi_{\pi_j} \simeq U_j(X),\ \ \ j \in \mathbb{Z}_{\ge 0}. \nonumber
		\end{equation}
	\end{lem}
	\begin{proof}
		We have
		\begin{align}
			\chi_{\pi_j}\left(\left(\begin{array}{cc}
				e^{i\theta} & 0 \\
				0 & e^{-i\theta}
			\end{array} \right) \right) &= \frac{\sin{((j+1)\theta )}}{\sin{\theta}} = U_j(\cos{\theta}),\ \ \ \theta \in \mathbb{R} \nonumber \\
			&= U_j \left( \frac{1}{2} \chi_{\pi_1}\left(\left(\begin{array}{cc}
				e^{i\theta} & 0 \\
				0 & e^{-i\theta}
			\end{array} \right) \right) \right). \nonumber
		\end{align}
		Thus, we obtain the stated result. \\
	\end{proof}  \vspace{3mm}
	
	Then \(\pi_{k+1} \) corresponds to the second kind of Chebyshev polynomial \(U_{k+1} = (k+1)^{-1} dT_{k+2}/dX \). Thus, we can identify
	\begin{equation}
		R_k \simeq R({\rm SU}_2)/I_k({\rm SU}_2), \nonumber
	\end{equation}
	This identification was used by Gepner (\cite{G1991}). Thus, our special basis \(\{[U_j] \}_{j=0}^{k} \) is given by the symmetric powers of the standard representations \(\{\pi_j \}_{j=0}^{k} \). \vspace{3mm} 
	
	We give a relation between representations of \({\rm SU}_2 \) and holomorphic data. Let \((\rho,V) \) be a representation of \({\rm SU}_2 \) and \(V = \bigoplus_{j \in \mathbb{Z}} V(j) \) be the weight decomposition, where
	\begin{equation}
		V(j) = \left\{v \in V\ | \ \rho\left( \left(\begin{array}{cc}
			e^{\sqrt{-1} \theta} & 0 \\
			0 & e^{-\sqrt{-1}\theta}
		\end{array} \right) \right)v = e^{\sqrt{-1} \theta j }v, \ \ \ \forall \theta \in \mathbb{R} \right\}. \nonumber
	\end{equation}
	In this paper, we call \(j \in \mathbb{Z} \) weight of \((\rho,V) \) if \(V(j) \neq \emptyset \) and highest weight of \((\rho,V) \) if \(V(j) \neq \emptyset \) and \(V(l) = \emptyset \) for all \(l \ge j+1 \).  \vspace{3mm}
	
	\begin{proposition} \label{prop4.1}
		There is a one-to-one correspondence between
		\begin{enumerate}
			\item [(a)] an equivalence class of representation \((\rho,V )\) of \({\rm SU}_2 \) with highest weight \(\le k \) and weight decomposition \( V = \bigoplus_{j \in \mathbb{Z}} V(j) \).
			\vspace{2mm}
			
			\item [(b)] a holomorphic potential
			\begin{equation}
				\xi = \frac{1}{\lambda} \left(\begin{array}{ccc}
					& & z^{l_0} \\
					& \iddots & \\
					z^{l_k} & &
				\end{array} \right)dz, \nonumber
			\end{equation}
			on \(\mathbb{C} \), where \(l_0,\cdots,l_k \in \mathbb{Z}_{\ge 0} \).
		\end{enumerate}
		Here, the correspondence is given by \(l_j = {\rm dim}V(j) - {\rm dim}V(j+2) \).
	\end{proposition}
	\begin{proof}
		Given a representation \((\rho,V) \) of \({\rm SU}_2 \) with highest weight \(\le k \), we show \(l_0,\cdots,l_k \ge 0 \). It is well-known that there exist \(l_0,\cdots,l_k \in \mathbb{Z}_{\ge 0} \) such that
		\begin{equation}
			(\rho,V) \simeq \bigoplus_{j=0}^{k} (\pi_j,V_j)^{\oplus l_j }, \nonumber
		\end{equation}
		where \((\pi_j,V_j) \) is the \(j \)-th symmetric power of the standard representation. Let \(V_j = \bigoplus_{a=0}^{j} V_j(j-2a) \) be the weight decomposition of \(V_j \), then we have \({\rm dim}V_j(j-2a) = 1 \) and
		\begin{align}
			V(2m) = \bigoplus_{b = m}^{[\frac{k}{2}] } V_{2b}(2m)^{\oplus l_{2b} },\ \ \ V(2m+1) = \bigoplus_{b = m}^{[\frac{k-1}{2}] } V_{2b+1}(2m+1)^{\oplus l_{2b+1} }, \nonumber
		\end{align}
		\begin{equation}
			V(j) = \left\{
			\begin{array}{cc}
				\bigoplus_{b = m}^{[\frac{k}{2}] } V_{2b}(2m)^{\oplus l_{2b} } & {\rm if}\ j = 2m, \\
				\bigoplus_{b = m}^{[\frac{k-1}{2}] } V_{2b+1}(2m+1)^{\oplus l_{2b+1} } & {\rm if}\ j = 2m+1,
			\end{array}
			\right. \nonumber
		\end{equation}
		for \(0 \le j \le k\). Thus, we obtain \(l_j = {\rm dim}V(j) - {\rm dim}V(j+2) \ge 0 \). \vspace{1mm} \\
		Conversely, given a holomrphic potential \(\xi \). We put
		\begin{equation}
			(\rho, V) = \bigoplus_{j=0}^{k} (\pi_j,V_j)^{l_j}, \nonumber
		\end{equation}
		then \((\rho,V) \) is a representation of \({\rm SU}_2 \) with highest weight \(\le k \), weight decomposition \(V=\bigoplus_{j \in \mathbb{Z}} V(j) \) and \(l_j = V(j) - V(j+2) \). \\
	\end{proof}  \vspace{3mm}
	
	From a holomorphic potential in Proposition \ref{prop4.1}, we obtain a global solution \(\{w_j\}_{j=0}^{k} \) of (A) with the asymptotic behaviour
	\begin{equation}
		w_j \sim \frac{l_j - l_{k-j}}{l_j + l_{k-j} + 2} \log|t| \ \ \ {\rm as}\ \ t = \frac{2}{l_j+l_{k-j}+2 } z^{\frac{l_j+l_{k-j}+2}{2} } \rightarrow 0. \nonumber
	\end{equation}
	Thus, we obtain the following Corollary.  \vspace{3mm}
	
	\begin{cor}\label{cor4.2}
		Fix \(n \in \mathbb{Z}_{\ge 0} \). There is a one-to-one correspondence between
		\begin{enumerate}
			\item [(a')] an equivalence class of representations \((\rho,V )\) of \({\rm SU}_2 \) with highest weight \(\le k \) and weight decomposition \( V = \bigoplus_{j \in \mathbb{Z}} V(j) \) such that
			\begin{equation}
				n = {\rm dim}V(j) - {\rm dim}V(j+2) + {\rm dim}V(k-j) - {\rm dim}V(k-j+2), \nonumber
			\end{equation}
			for \(j = 0,\cdots,k \),
			\vspace{2mm}
			
			\item [(b')] a solution \(\{w_j \}_{j=0}^{k} \) of (A) with the asymptotic data
			\begin{equation}
				(m_0,m_1,\cdots,m_k) = \left(\frac{2l_0 - n}{n+2},\frac{2l_1 - n}{n+2},\cdots,\frac{2l_k - n}{n+2} \right). \nonumber
			\end{equation}
		\end{enumerate}
	\end{cor}
	Thus, we obtain a relation between certain solutions of the tt*-equation constructed from the \(({\rm SU}_2)_k \)-fusion ring, holomorphic data and representations of \({\rm SU}_2 \).

	\subsection{A particular solution, and the supersymmetric \(A_k \) minimal model}
	Let us consider the particular solution of (A) given by the choice \(l_j = j \) and \(n=k \), and its corresponding holomorphic potential. By Proposition 3.3 (and its proof), this solution corresponds to the \(k+1 \)-solutions of the sinh-Gordon equaion which have holomorphic potentials \(\xi_j\ (1 \le j \le n) \). \vspace{2mm}
	
	Through the procedure of Corollary 4.5 of \cite{GO2022}, these holomorphic potentials correspond to the weights of level \(k \) in the Fundamental Weyl Alcove of \({\rm SU}_2 \), or, more naturally, to the positive energy representations of the affine Kac-Moody group (loop group) of \({\rm SU}_2 \) with highest weights \((j,k) \) (where \(k \) is the level). These are precisely the representations in the \(A_{k+1} \) minimal model, or more precisely the Virasoro minimal model of type \((2,k) \). We refer to \cite{GO2022} for the details of this correspondence, and to \cite{BS1993}, \cite{DMS1997} for Virasoro minimal models.  \vspace{2mm}
	
	We note that the asymptotic data from Corollary \ref{cor4.2} is
	\begin{equation}
		(m_0,\cdots,m_k) = \left(\frac{-k}{k+2}, \frac{-k+2}{k+2}, \cdots, \frac{k}{k+2} \right). \nonumber
	\end{equation}
	The same data arises as follows, as was observed by Cecotti and Vafa in \cite{C1991}. The supersymmetric \(A_k \) minimal model can be described by irreducible unitary highest weight representations of the Ramond algebra, an \(N=2 \) super Virasoro algebra. The \(N=2 \) super Virasoro algebra is an extension of Virasoro algebra spaned by certain operators \(\{{\bf 1},L_r,J_r,G_r^+,G_r^-\}_{r \in \mathbb{Z}} \) with certain commutation relations \cite{KV1989}. The highest weight representations can be characterized by a nonzero vector \(v \) satisfying 
	\begin{align}
		&L_r v = J_r v = G_r^{\pm}v = 0,\ \ \ \forall r >0, \nonumber \\
		&L_0 v = h v,\ \ \ J_0 v = Q v. \nonumber
	\end{align}
	and the eigenvalues \(h, Q \) take a finite number of values \cite{BFK1986}. In particular, the values of \(Q \) are given by
	\begin{equation}
		Q = \frac{-k}{k+2}, \frac{-k+2}{k+2}, \cdots, \frac{k}{k+2}. \nonumber
	\end{equation} \vspace{2mm}
	
	A further relation with physics is provided by the Landau-Ginzburg model based on the \(A_{k+1} \) singularity. It is well-known from \cite{CV1991} that this corresponds to the solution of the tt*-Toda equation with exactly the above asymptotic data \(m= (m_0,\cdots,m_k) \). Here we use the notation of \cite{GIL2023} for asymptotic data of the tt*-Toda equation, where the existence of this solution was established. The solution of (A) corresponding to this data, and the solution of the tt*-Toda equation corresponding to this data, are quite different functions (they are solutions of different equations). However, the coincidence of their asymptotic data is valid not only for the first term, but also for the respective second terms, which involve the gamma function expressions \(\Gamma(\frac{1}{k+2}), \Gamma(\frac{2}{k+2}),\cdots,\Gamma(\frac{k+1}{k+2}) \). This can be seen by substituting the above value of \(m = (m_0,\cdots,m_k) \) into the expression for \(w_j \) in Corollary 8.14 of \cite{GIL2023}: precisely these gamma function expressions are obtained. \vspace{2mm}
	
	As pointed out by Cecotti and Vafa, it is the first and second terms (rather than the first terms alone) which represent the \lq\lq boundary data\rq\rq of the solution of the tt*-equation. This deeper coincidence of the boundary data for our solution of (A) (given by the \(({\rm SU}_2)_k \)-fusion algebra) with the boundary data for the \(A_{k+1} \) singularity solution of the tt*-Toda equation (given by the \(({\rm SU}_{k+1} )_1 \)-fusion algebra) may be a reflection of the physically expected \lq\lq level-rank duality\rq\rq. \vspace{3mm}
	
	Let \((E,\eta,C,g) \) be the tt*-structure in section 2 with holomorphic data \((l_0,l_1,\cdots,l_k) = (0,1,\cdots,k) \). For representation \((\rho,V) \) of \({\rm SU}_2 \), we define a holomorphic section \(\tau_{\rho} \) by \(\tau_{\rho}(t) = (t,\sqrt{2(k+2)}[P(\rho)] ) \), where \(P \) is the projection from representations of \({\rm SU}_2 \) into \(R({\rm SU}_2)/I_k({\rm SU}_2) \) . Then we obtain the following Proposition.  \vspace{3mm}
	
	\begin{proposition} \label{prop4.3}
		If \((\rho,V) \) is an irreducible representation of \({\rm SU}_2 \), then \(w_{\rho} = g(\tau_{\rho},\tau_{\rho} ) \) is a solution of the sinh-Gordon equation. \vspace{1mm} \\
		Furthermore, let \(l_{\rho } \) be the highest weight of representation \((\rho,V) \) then we have
		\begin{equation}
			w_{\rho} = w_{\tilde{\rho}}\ \ \ \Longleftrightarrow \ \ \ \frac{l_{\rho}-l_{\tilde{\rho} }}{k+2} \in 2\mathbb{Z}\ \ {\rm or}\ \ \frac{l_{\rho}+l_{\tilde{\rho} }-k}{k+2} \in 2\mathbb{Z}+1. \nonumber 
		\end{equation}
	\end{proposition} \vspace{3mm}
	
	To prove Proposition \ref{prop4.3}, we use the following Lemma. \vspace{3mm}
	
	\begin{lem}\label{lem4.2}
		For \(\alpha \in \mathbb{Z}_{\ge 0}, 0 \le \beta \le k+1 \), we have
		\begin{equation}
			[U_{\alpha(k+2)+\beta}] = \left\{\begin{array}{ccc}
				0 & {\rm if}\ \beta = k+1, \nonumber \\
				\lbrack U_{\beta} \rbrack & {\rm if}\ \alpha:{\rm even },\ 0 \le \beta \le k. \nonumber \\
				-[U_{k-\beta}] & {\rm if}\ \alpha:{\rm odd },\ 0 \le \beta \le k. \nonumber
			\end{array}
			\right.
		\end{equation}
	\end{lem}
	\begin{proof}
		We put \(j = \alpha(k+2) + \beta,\ \alpha \in \mathbb{Z}_{\ge 0},\ 0 \le \beta \le k+1 \). For \(l = 1,\cdots,k+1 \), we have
		\begin{equation}
			U_{k+1}\left(\cos{\left(\frac{l}{k+2}\pi \right)} \right) = \frac{\sin{\left(\frac{k+2}{k+2} l\pi \right)}}{\sin{\left(\frac{l}{k+2}\pi \right) }} = 0, \nonumber
		\end{equation}
		and if \(\alpha \) is even, \(0 \le \beta \le k \) then we have
		\begin{align}
			U_j\left(\cos{\left(\frac{l}{k+2}\pi \right)} \right) &=  \frac{\sin{\left(\frac{j+1}{k+2} l\pi \right)}}{\sin{\left(\frac{l}{k+2}\pi \right) }} = \frac{\sin{\left(\frac{\alpha(k+2)+\beta+1}{k+2} l\pi \right)}}{\sin{\left(\frac{l}{k+2}\pi \right) }} \nonumber \\
			&=  \frac{\sin{\left(\frac{\beta+1}{k+2} l\pi \right)}}{\sin{\left(\frac{l}{k+2}\pi \right) }} \nonumber \\
			&= U_{\beta}\left(\cos{\left(\frac{l}{k+2}\pi \right)} \right), \nonumber
		\end{align}
		if \(\alpha \) is odd, \(0 \le \beta \le k \) then we have
		\begin{align}
			U_j\left(\cos{\left(\frac{l}{k+2}\pi \right)} \right) &= \frac{\sin{\left(\frac{\alpha(k+2)+\beta+1}{k+2} l\pi \right)}}{\sin{\left(\frac{l}{k+2}\pi \right) }} \nonumber \\
			&= -\frac{\sin{\left(-\frac{\alpha(k+2)+\beta+1}{k+2} l\pi + (\alpha + 1)l\pi \right)}}{\sin{\left(\frac{l}{k+2}\pi \right) }} \nonumber \\
			&= -\frac{\sin{\left(\frac{k-\beta+1}{k+2} l\pi \right)}}{\sin{\left(\frac{l}{k+2}\pi \right) }} \nonumber \\
			&= -U_{k-\beta}\left(\cos{\left(\frac{l}{k+2}\pi \right)} \right). \nonumber
		\end{align}
		From Lemma \ref{lem2.2}, we obtain the stated result. \\
	\end{proof} \vspace{3mm}

	\begin{proof}[Proof of Proposition \ref{prop4.3}]
		Since \(\rho \) is irreducible, \(\rho \) is isomorphic to \(\pi_{l_[\rho] } \) and then \(P(\rho) = [U_{l_{\rho}}] \). From Lemma \ref{lem4.2}, for \(l_{\rho} = \alpha(k+2)+\beta,\ \alpha \in \mathbb{Z}_{\ge 0}, 0 \le \beta \le k+1 \) we have
		\begin{equation}
			[U_{l_{\rho}}] = \left\{\begin{array}{ccc}
				0 & {\rm if}\ \beta = k+1, \nonumber \\
				\lbrack U_{\beta} \rbrack & {\rm if}\ \alpha:{\rm even },\ 0 \le \beta \le k. \nonumber \\
				-[U_{k-\beta}] & {\rm if}\ \alpha:{\rm odd },\ 0 \le \beta \le k. \nonumber
			\end{array}
			\right.
		\end{equation}
		Thus, \(w_{\rho} \) is a solution of the sinh-Gordon equation. \vspace{1mm} \\
		Next, We put \(l_{\rho} = \alpha(k+2)+\beta, l_{\tilde{\rho}} = \tilde{\alpha}(k+2)+\tilde{\beta}\), where \(\alpha, \tilde{\alpha} \in \mathbb{Z}_{\ge 0},\ 0 \le \beta,\ \tilde{\beta} \le k+1 \). Suppose that \(w_{\rho} = w_{\tilde{\rho}} \). Since holomorphic data of \(g \) is given by \((l_0,\cdots,l_k) = (0,\cdots,k) \), from Lemma \ref{lem4.2} we see that \(\alpha,\beta,\tilde{\alpha},\tilde{\beta} \) satisfy \(\alpha - \tilde{\alpha} \in 2\mathbb{Z},\ \beta = \tilde{\beta} \) or \(\alpha + \tilde{\alpha} \in 2 \mathbb{Z}+1,\ \beta = k - \tilde{\beta} \). Thus, we obtain
		\begin{equation}
			\frac{l_{\rho}-l_{\tilde{\rho}}}{k+2} \in 2\mathbb{Z},\ \ {\rm or}\ \ \frac{l_{\rho}+l_{\tilde{\rho}}-k}{k+2} \in 2\mathbb{Z}+1. \nonumber 
		\end{equation}
		Conversely, suppose that \(l_{\rho},l_{\tilde{\rho}} \) satisfy the condition above. Then we have
		\begin{equation}
			\alpha - \tilde{\alpha} + \frac{\beta-\tilde{\beta}}{k+2} \in 2\mathbb{Z} \ \ {\rm or}\ \ \alpha + \tilde{\alpha} + \frac{\beta+\tilde{\beta}-k}{k+2} \in 2\mathbb{Z}+1, \nonumber
		\end{equation}
		and then \(\alpha - \tilde{\alpha} \in 2\mathbb{Z},\ \beta = \tilde{\beta} \) or \(\alpha + \tilde{\alpha} \in 2 \mathbb{Z}+1,\ \beta = k - \tilde{\beta} \). Since holomorphic data of \(g \) is given by \((l_0,\cdots,l_k) = (0,\cdots,k) \), we obtain \(w_{\rho} = w_{\tilde{\rho}} \). \\
	\end{proof} \vspace{3mm}
	
	In this subsection, we give a relation between irreducible representations and Smyth potentials. It can be shown in my unpublished paper that we can use the following Theorem for the classification of Smyth potentials. \vspace{3mm}
	
	\begin{theorem}[T.Udagawa, in preparation] \label{thm4.4}
		Let
		\begin{equation}
			\xi_j = \frac{1}{\lambda} \left(\begin{array}{cc}
				0 & z^j \\
				z^{k-j} & 0
			\end{array} \right)dz,\ \ \ \xi_j = \frac{1}{\lambda} \left(\begin{array}{cc}
				0 & z^l \\
				z^{k-l} & 0
			\end{array} \right)dz,\ \ \ j,l \in \mathbb{Z}, \nonumber
		\end{equation}
		be Smyth potentials on \(\mathbb{C} - (-\infty,0] \). Then
		\begin{equation}
			\exists C \in \Lambda^+ {\rm SL}_2 \mathbb{C}\ \ \ {\rm s.t.}\ \ \ \xi_l = \xi_j \cdot C, \nonumber
		\end{equation}
		if and only if \(j,l \) satisfy the condition
		\begin{equation}
			\frac{j-l}{k+2} \in 2\mathbb{Z}\ \ {\rm or}\ \ \frac{j+l-k}{k+2} \in 2\mathbb{Z}+1. \nonumber
		\end{equation}
	\end{theorem}
	\begin{proof}[Sketch of Proof]
		We put
		\begin{equation}
			C_+= \left(
			\begin{array}{cc}
				z^{j-k-1 } & -\lambda (j-k-1 ) \\ 0 & z^{-j+k+1}
			\end{array}
			\right),
			\ C_- = \left(
			\begin{array}{cc}
				z^{j+1 }  & 0 \\   -\lambda (j+1)  & z^{-j-1}
			\end{array}
			\right), \nonumber
		\end{equation}
		then we have
		\begin{align}
			&\xi_{j+2m(k+2) } = \xi_j \cdot (\underbrace{C_-C_+ \cdots C_-C_+}_{2m} ), \nonumber\\
			&\xi_{j-2m(k+2) } = \xi_j \cdot (\underbrace{C_+C_- \cdots C_+C_-}_{2m} ), \nonumber\\
			&\xi_{-j+k-(2m+1)(k+2) } = \xi_j \cdot (\underbrace{C_-C_+ \cdots C_-C_+}_{2m}C_- ),\nonumber\\
			&\xi_{-j+k+(2m+1)(k+2) } = \xi_j \cdot (\underbrace{C_+C_- \cdots C_+C_-}_{2(m+1)}C_+ ),\ \ \ \ \ m \in \mathbb{N} \cup \{0\}. \nonumber
		\end{align}
		Thus, if \(j,l \) satisfy
		\begin{equation}
			\frac{j-l}{k+2} \in 2\mathbb{Z}\ \ {\rm or}\ \ \frac{j+l-k}{k+2} \in 2\mathbb{Z}+1, \nonumber
		\end{equation}
		then
		\begin{equation}
			\exists C \in \Lambda^+ {\rm SL}_2 \mathbb{C}\ \ \ {\rm s.t.}\ \ \ \xi_l = \xi_j \cdot C. \nonumber
		\end{equation}
		The inverse is given by investigating the condition on \(k,l \) when \(C = \sum_{i=0}^{\infty} C_j(z) \lambda^i \) converges for all \(\lambda \in S^1 \). \\
	\end{proof} \vspace{3mm}
	
	From Proposition \ref{prop4.3} and Theorem \ref{thm4.4}, we obtain the following Theorem. \vspace{3mm}\\
	
	\begin{theorem}\label{thm4.5}
		Fix \(k \in \mathbb{N} \). We consider the correspondence between
		\begin{enumerate}
			\item [(A)] equivalence classes of irreducible representations \((\rho,V) \) of \({\rm SU}_2 \), \vspace{2mm}
			
			\item [(B)] Smyth potentials
			\begin{equation}
				\xi_{\rho} = \frac{1}{\lambda} \left(\begin{array}{cc}
					0 & z^{l_{\rho}} \\
					z^{k-l_{\rho} } & 0
				\end{array} \right)dz, \nonumber
			\end{equation}
			where \(l_{\rho} \) is highest weight of \((\rho,V) \).
		\end{enumerate}
		Then we have
		\begin{equation}
			w_{\rho } = w_{\tilde{\rho} }\ \ \ \Longleftrightarrow \ \ \ \exists C \in \Lambda^+ {\rm SL}_2 \mathbb{C}\ \ {\rm s.t.}\ \ \xi_{\tilde{\rho}} = \xi_{\rho} \cdot C. \nonumber
		\end{equation}
	\end{theorem} \vspace{3mm}
	
	As a Corollary, we obtain the following relation. \vspace{3mm}
	
	\begin{cor} \label{cor4.5}
		Fix \(k \in \mathbb{N} \). There is a one-to-one correspondence between
		\begin{enumerate}
			\item [(i)] solutions \(w \) of the sinh-Gordon equation with the asymptotic data
			\begin{equation}
				\left\{m = \frac{j+1}{k+2}\ | \ j = 0,\cdots,k \right\}, \nonumber
			\end{equation}
			
			\item [(ii)] equivalence classes of irreducible representations \((\rho,V) \) with highest weight \(j \le k \), \vspace{1mm}
			
			\item [(iii)] equivalence classes of Smyth potentials
			\begin{equation}
				[\xi_j]_+ = \left[\frac{1}{\lambda} \left(\begin{array}{cc}
					0 & z^j \\
					z^{k-j} & 0
				\end{array} \right)dz \right]_+, \nonumber
			\end{equation}
			where \([\cdot]_+ \) denotes the equivalence class by gauging action of \(\Lambda^+ {\rm SL}_2 \mathbb{C} \).
		\end{enumerate}
	\end{cor}
	\begin{proof}
		We obtain \((i) \Leftrightarrow (ii) \) from Proposition \ref{prop3.3} and Theorem \ref{thm4.4}. We obtain \((ii) \Leftrightarrow (iii) \) from Theorem \ref{thm4.4} and Theorem \ref{thm4.5}. \\
	\end{proof} \vspace{3mm}
	
	Thus, we obtain the relation between solutions, irreducible representations and Smyth potentials. From Corollary \ref{cor4.5}, the holomorphic data \(l_j = j \) can be interpreted as irreducible representations with highest weight \(j \). In physics, irreducible representations can be seen as vacua in the \(N=2 \) sinh-Gordon model (see \cite{CV19932}). \vspace{2mm}
	
	We remark that our analysis for the \(({\rm SU}_2)_k \)-fusion ring can be generalized to the \(({\rm SU}_n)_k \)-fusion algebra for any \(n \ge 2 \). This more general situation was also considered by Cecotti and Vafa (\cite{CV1991}). In this case, a solution of the equation corresponds to a finite number of solutions of the tt*-Toda equation.

	
	
	
	
	\subsection*{Acknowledgement}
	The author would like to thank Professor Martin Guest for his considerable support and thank Professor Takashi Otofuji for useful conversations.
	
	
	\bibliography{mybibfile}

\begin{thebibliography}{10}

\bibitem{BI1995}
B.~Bobenko and A.~Its.
\newblock {The Painlev\'e III equation and the Iwasawa decomposition}.
\newblock {\em Manuscripta Math.}, 87(3):369--377, 1995.

\bibitem{BFK1986}
W.~Boucher, D.~Friedan, and A.~Kent.
\newblock {Determinant formula and unitarity for the $N=2$ superconformal
  algebras in two dimensions or exact results on string compatification}.
\newblock {\em Phys. Lett. B}, 172(3--4):316--322, 1986.

\bibitem{BS1993}
P.~Bouwknegt and K.~Schoutens.
\newblock {W-symmetry in conformal field theory}.
\newblock {\em Phys. Rep.}, 223(4):183--276, 1993.

\bibitem{C1991}
S.~Cecotti.
\newblock {Geometry of $N=2$ Landau-Ginzburg families}.
\newblock {\em Nuclear Phys. B}, 355(3):755--775, 1991.

\bibitem{CV1991}
S.~Cecotti and C.~Vafa.
\newblock Topological-anti-topological fusion.
\newblock {\em Nuclear Phys. B}, 367(2):359--461, 1991.

\bibitem{CV19932}
S.~Cecotti and C.~Vafa.
\newblock {Ising model and $N=2$ supersymmetric theories}.
\newblock {\em Comm. Math. Phys.}, 157(1):139--178, 1993.

\bibitem{CV1993}
S.~Cecotti and C.~Vafa.
\newblock {On classification of N = 2 supersymmetric theories}.
\newblock {\em Comm. Math. Phys.}, 158:569--644, 1993.

\bibitem{D1993}
B.~Dubrovin.
\newblock {Geometry and integrability of topological-antitopological fusion}.
\newblock {\em Comm. Math. Phys.}, 152:539--564, 1993.

\bibitem{FIKN2006}
A.~Fokas, A.~Its, A.~Kapaev, and V.~Novokshenov.
\newblock {\em {Painlev\'e Transcendents: The Riemann-Hilbert Approach,
  Mathematical}}.
\newblock Surveys and Monographs 128. Amer. Math Soc., 2006.

\bibitem{DMS1997}
P.~Di Francesco, P.~Mathieu, and D.~S\'en\'echal.
\newblock {\em {Conformal field theory}}.
\newblock Grad. Texts Contemp. Phys. Springer, 1997.

\bibitem{G1991}
D.~Gepner.
\newblock {Fusion rings and geometry}.
\newblock {\em Comm. Math. Phys.}, 141(2):381--411, 1991.

\bibitem{GIL20151}
M.~Guest, A.~Its, and C.~Lin.
\newblock Isomonodromy aspects of the tt* equations of Cecotti and Vafa I.
  Stokes data.
\newblock {\em Int. Math. Res. Notices}, 2015(22):11745--11784, 2015.

\bibitem{GIL20152}
M.~Guest, A.~Its, and C.~Lin.
\newblock {Isomonodromy aspects of the tt* equations of Cecotti and Vafa II.
  Riemann-Hilbert problem}.
\newblock {\em Comm. Math. Phys.}, 336:337--380, 2015.

\bibitem{GIL2020}
M.~Guest, A.~Its, and C.~Lin.
\newblock {Isomonodromy aspects of the tt* equations of Cecotti and Vafa III.
  Iwasawa factorization and asymptotics}.
\newblock {\em Comm. Math. Phys.}, 374(2):923--973, 2020.

\bibitem{GIL2023}
M.~Guest, A.~Its, and C.~Lin.
\newblock {The tt*-Toda equations of $A_n$ type}.
\newblock arXiv, 2023, https://doi.org/arXiv:2302.04597.

\bibitem{GL2014}
M.~Guest and C.~Lin.
\newblock {Nonlinear PDE aspects of the tt* equations of Cecotti and Vafa}.
\newblock {\em J. Reine Angew. Math.}, 689:1--32, 2014.

\bibitem{GO2022}
M.~Guest and T.~Otofuji.
\newblock {Positive energy representations of affine algebras and Stokes
  matrices of the affine Toda equations}.
\newblock {\em Adv. Theor. Math. Phys.}, 26(7):2077--2093, 2022.

\bibitem{KV1989}
V.~Kac and J.~van~de Leur.
\newblock {On classification of superconformal algebras}.
\newblock {\em Strings '88}, World Scientific Publishing, 77--106, 1989.

\bibitem{MTW1977}
B.~McCoy, C.Tracy, and T.~Wu.
\newblock {Painle\'e functions of the third kind}.
\newblock {\em J. Math. Phys.}, 18(5):1058--1092, 1977.

\end{thebibliography}
	\bibliographystyle{plain}

\end{document}